\documentclass[11pt]{amsart}
\usepackage{amsmath}
\usepackage{amssymb}
\usepackage{amsthm}
\usepackage{amsfonts}
\usepackage{mathrsfs}
\usepackage{latexsym}
\usepackage{hyperref}
\usepackage{enumerate}
\usepackage[all]{xy}
\usepackage{color}
\usepackage{comment}
\usepackage{slashed}

\usepackage{amssymb,amsmath,amscd,url,geometry,pdflscape}
\usepackage{verbatim,url}

\usepackage{marginnote}
\usepackage{comment} 

\newtheorem{thm}{Theorem}[section]
\newtheorem{cor}[thm]{Corollary}
\newtheorem{lem}[thm]{Lemma}

\newtheorem{defn}[thm]{Definition}

\def\P{\mathbb{P}}

\def\fB{\mathfrak{B}}

\def\fM{\mathfrak{M}}

\def\fS{\mathfrak{S}}

\def\SU{\mathrm{SU}}

\def\GL{\mathrm{GL}}

\def\R{\mathbb{R}}
\def\C{\mathbb{C}}

\def\Z{\mathbb{Z}}

\def\SL{\mathrm{SL}}

\begin{document}

\title[$2$-Spinors]{$2$-Spinors via Linear Algebra}

\author[R. Plymen]{Roger Plymen}
\address{School of Mathematics, Southampton University, Southampton SO17 1BJ,  England
\emph{and} School of Mathematics, Manchester University, Manchester M13 9PL, England}
\email{r.j.plymen@soton.ac.uk \quad roger.j.plymen@manchester.ac.uk}

\maketitle

\begin{abstract}   We give a streamlined account of $2$-spinors, up to and including the Dirac equation, using little more than the resources of linear algebra.   
We prove that the Dirac bundle is isomorphic to the associated bundles $\SL_2(\C) \times_{\SU_2} S$ and $\SL_2(\C) \times_{\SU_2} \overline{S}$.  
A solution of the Dirac equation  determines a pair of conjugate $2$-spinor fields over the mass shell $X_m$.   
\end{abstract} 

\tableofcontents

\section{Introduction} The classic work on spinors and space-time is by Penrose and Rindler \cite{PR}.   In this magisterial work, the authors demonstrate that 
$2$-spinors are woven into the very fabric of space-time.   This point is further made in the books by O'Donnell \cite{O} and Wald [Chapter 13]\cite{W}.

It is our aim in this Note to show that the basic theory of $2$-spinors, up
to and including the Dirac equation, can be formulated using little more than the resources of linear algebra.    

Throughout this Note, we define the vector space of $4$-spinors as
\[
\fS = S \oplus \overline{S}
\]
where $S$ is the vector space of $2$-spinors and $\overline{S}$ is the conjugate vector space of $S$.

Our approach is canonical except at the moment in \S 2.2  when one has to identify a real finite-dimensional vector space $V$ with its space of characters.   
 
 The Dirac bundle is defined and studied in some detail in the the books by Simms \cite{Sim} and  Varadarajan \cite{Var}.   
 Our new result is that the Dirac bundle  as an $\SL_2(\C)$-bundle is isomorphic to the associated bundles
\begin{eqnarray}\label{S1}
\SL_2(\C) \times_{\SU_2} S
\end{eqnarray}
and
\begin{eqnarray}\label{S2}
\SL_2(\C) \times_{\SU_2} \overline{S}
\end{eqnarray}
where $\SU_2$ acts on $S$ and $\overline{S}$ as the spin $1/2$ representation.   As a consequence,  a solution of the 
Dirac equation determines a pair of conjugate $2$-spinor fields over the mass shell $X_m$.

 In section 2, we present the basic theory of $2$-spinors via linear algebra.    In section 3, we define the Dirac bundle and prove that it is isomorphic to the associated vector bundles (\ref{S1})
 and (\ref{S2}).       In section 4, we briefly discuss the transition from the vector space of $2$-spinors to spinor fields on space-time.

\section{$2$-spinors}\label{SS}   Let $W$ be a complex vector space.   Then $W$ comprises an abelian group $A$ and a scalar-multiplication map 
\begin{eqnarray}\label{m1}
m : (\C,A) \to A, \quad \quad ( \lambda, v) \mapsto  \lambda v
\end{eqnarray}

Each complex vector space $W$ has a companion vector space,  the \emph{conjugate vector space} denoted $\overline{W}$.   Now $W$ and $\overline{W}$ share the same
underlying abelian group $A$, but the scalar-multiplication map for $\overline{W}$ is given by
\begin{eqnarray}\label{m2}
\overline{m} : (\C,A) \to A, \quad \quad (\lambda,v) \mapsto \overline{\lambda}v
\end{eqnarray}

Let $v \in W$.   This means that $v \in A$ equipped with the scalar-multiplication (\ref{m1}).  

The vector spaces $W$ and $\overline{W}$ enjoy a perfect duality, in the sense that each is the conjugate of the other.   

 If we wish to refer to $v \in A$ equipped with the scalar-multiplication (\ref{m2}) we shall write
$\overline{v}$ instead of $v$.   From a logical point of view, we have $v = \overline{v} \in A$ but $v$ and $\overline{v}$ lie in the distinct vector spaces $W$ and $\overline{W}$.   
From a practical point of view, it is good to adopt the following convention:   we shall always write 
\[
v \in W \quad \textrm{and} \quad \overline {v} \in \overline{W}
\]
In that case, in accordance with (\ref{m2}), we have
\[
\overline{\lambda v} = \overline{\lambda} \overline{v}
\]

If $W$ is finite-dimensional with basis $\{e_1, \ldots, e_n\}$ then $\overline{W}$ admits the basis$\{\overline{e_1}, \ldots, \overline{e_n})$ so that 
$W$ and $\overline{W}$ have the same dimension.    

Given $A \in \GL(W)$ define the \emph{conjugate} of $A$ as follows:
\[
\overline{A} \overline{v} : = \overline{Av}
\]
Then we have
\begin{eqnarray*}
\overline{A}(\lambda \overline{v} + \mu \overline{w}) &=&  \overline{A(\overline{\lambda}v + \overline{\mu}w})\\
&=& \overline{\overline{\lambda} Av + \overline{\mu}Aw}\\
&=& \lambda \overline{Av} + \mu \overline{Aw}\\
&=& \lambda  \overline{A}\overline{v} + \mu \overline{A} \overline{w}
\end{eqnarray*}
so that $\overline{A}$ is a \emph{linear} map $\overline{W} \to \overline{W}$.   We have $\overline{A} \in \GL(\overline{W})$.   
The matrix of the conjugate of $A$  is the conjugate of the matrix of $A$.

  Now let $S$ denote a complex vector space of dimension $2$.   Consider the tensor product over $\C$ 
\[
S \otimes \overline{S}
\]
This is a complex vector space of dimension $4$.   Note the standard tensor product rule:
\[
\lambda(x \otimes \overline{y}) = (\lambda x) \otimes \overline{y} = x \otimes \overline{\lambda} \overline{y}
\]

This vector space admits a canonical involution $J$ defined on elementary tensors as follows:
\[
J : S \otimes \overline{S} \to S \otimes \overline{S}, \quad \quad x \otimes \overline{y} \mapsto y \otimes \overline{x}
\]

We have $J^2 = I$ the identity map on $S \otimes \overline{S}$.   
The map $J$ admits two eigenvalues, namely $+1$ and $-1$.   

The $+1$-eigenspace of $J$ will be denoted $V$.  We define
\[
V: = \{v \in S \otimes \overline{S} : Jv = v\}
\]
Then $V$ is a canonical subspace of $S \otimes \overline{S}$.   We will view $V$ as a real vector space of dimension $4$.

Let $\bigwedge^2S$ denote the exterior square of $S$.   Choose a basis $\{e_1, e_2\}$ for $S$.   Then the $2$-vector $e_1 \wedge e_2$ is a basis for the $1$-dimensional vector space 
$\bigwedge^2S$.    Define an isomorphism 
$\varphi : \bigwedge^2S \simeq \C$ as follows:
\[
\varphi(\lambda e_1\wedge e_2) = \lambda
\]
Next, define
\[
\varepsilon(x,y): = \varphi(x \wedge y)
\]
Then we have
\begin{eqnarray*}
\varepsilon(x,y) &=& \varphi(x \wedge y)\\
 &=& \varphi(- y \wedge x)\\
 &=& - \varphi(y \wedge x)\\
 &=& -\varepsilon(x,y)
\end{eqnarray*}
and also
\[
\varepsilon(e_1,e_2) = \varphi(e_1 \wedge e_2) = 1
\]
so that $\varepsilon$ is a symplectic form on $S$ and $\{e_1, e_2\}$ is a symplectic basis, a \emph{dyad}.   

We emphasize that the symplectic form $\varepsilon$  \emph{arises from, and is determined by, a non-canonical choice of isomorphism}  
\[
\varphi : {\bigwedge}^2 S \simeq \C
\]

Define
\[
\SL(S): = \{ A \in \GL(S) : \det (A) = 1 \}.
\]
Once a basis in $S$ has been chosen, we have 
\[
\SL(S) \simeq \SL_2(\C)
\]
 Let $A \in \SL(S)$.      If $Ae_1 =  \alpha_{11}e_1 + \alpha_{12}e_2$ and $ Ae_2 =  \alpha_{21}e_1 + \alpha_{22}e_2$ then 
\begin{eqnarray*}
\varepsilon(Ae_1,Ae_2) &=& \varphi(Ae_1 \wedge Ae_2)  \\
&=&  \varphi(\det (A) e_1 \wedge e_2)  \\
&=& \varphi(e_1 \wedge e_2) \\
&=& \varepsilon(e_1, e_2)\\
&=& 1
\end{eqnarray*}
so that $\SL(S)$ acts freely and transitively on the set of all symplectic bases (dyads).

 The basis $\{\overline{e_1}, \overline{e_2}\}$ determines an isomorphism $\overline{\varphi} : \bigwedge^2 \overline{S} \simeq \C$ as follows:
\[
\overline{\varphi}(\lambda \overline{e_1} \wedge \overline{e_2}) = \lambda
\]
and allows us to define
\[
\overline{\varepsilon}(\overline{x}, \overline{y}): = \overline{\varphi}(\overline{x} \wedge \overline{y})
\]
as a symplectic form on $\overline{S}$.  

We have
\[
{\bigwedge}^2 \overline{S} = \overline{{\bigwedge}^2 S}
\]
and so
\[
\overline{\varepsilon}(\overline{x}, \overline{y}) = \overline{\varepsilon(x,y)}
\]

   We have 
\begin{eqnarray*}
   \overline{A} \overline{e_1} &=& \overline{Ae_1} \\
   &=& \overline{a_{11}e_1 + a_{12}e_2} \\
   &=& \overline{a_{11}} \; \overline{e_1} + \overline{a_{12}} \; \overline{e_2}
\end{eqnarray*}
and
\begin{eqnarray*}
   \overline{A} \overline{e_2} &=& \overline{Ae_2} \\
   &=& \overline{a_{21}e_1 + a_{22}e_2} \\
   &=& \overline{a_{21}} \; \overline{e_1} + \overline{a_{22}} \; \overline{e_2}
\end{eqnarray*}
so that
\[
\det(\overline{A}) = \overline{\det(A)} = 1
\]
and
\[
 \overline{A} \in \SL(\overline{S})
 \]

Then
\[
\overline{\varepsilon}(\overline{A} \overline{e_1}, \overline{A} \overline{e_2}) = 1
\]
as above and $\SL(\overline{S})$ acts freely and transitively on the set of dyads for $\overline{S}$.

\begin{defn}   We now define
\begin{eqnarray}\label{h}
h(a \otimes \overline{b}, c \otimes \overline{d}): &=& \varepsilon(a,c) \cdot \overline{\varepsilon}(\overline{b}, \overline{d})
\end{eqnarray}
\end{defn} 

This determines a bilinear form on $S \otimes \overline{S}$, thanks to the bilinearity of $\varepsilon$ and $\overline{\varepsilon}$.  
Furthermore, this is a \emph{symmetric} bilinear form, for $h$ is invariant under the  simultaneous interchanges $a \to c$ and $\overline{b} \to \overline{d}$.

We are especially interested in the restriction of $h$ to $V$.  We denote this restriction by $g$:
\[
g = h|_V
\] 
This will be a symmetric bilinear form on $V$, and so will have an associated quadratic form $Q$.    We next determine the rank and signature of this quadratic form.

Conceptually, $h$ is given by the following map:
\begin{eqnarray*}
S \otimes \overline{S} \otimes S \otimes \overline{S}  \cong S  \otimes S \otimes \overline{S} \otimes \overline{S} \to {\bigwedge}^2 S \otimes {\bigwedge}^2 \overline{S} \cong \C \otimes \C \cong \C
\end{eqnarray*}
Note that this map \emph{depends only on} $\varepsilon$.

\begin{thm} Let $x \in S \otimes \overline{S}$.    If $Jx = x$ then $g$, defined as 
\[
g(x): = h(x,x)
\]
is a quadratic form of rank $4$ and signature $2$, i.e. a Lorentz quadratic form.
\end{thm}

\begin{proof} We will do this via an explicit diagonalization of $g$.  
Let $\{e_1, e_2\}$ be a basis of $S$.  Then $e_i \otimes \overline{e_j}$ is a basis of $S \otimes \overline{S}$ with $1 \leq i,j \leq 2$.   We will  write
\[
e_i \cdot \overline{e_j}: = e_i \otimes \overline{e_j}
\]

Following \cite[3.1.20]{PR}, we define
\begin{eqnarray}\label{world}
u_0 &=& (e_1 \cdot \overline{e_1} + e_2 \cdot \overline{e_2})/\sqrt 2\\
u_1 &=& e_1 \cdot \overline{e_2} + e_2 \cdot \overline{e_1}/\sqrt 2 \nonumber\\
u_2 &=& i(e_1 \cdot \overline{e_2} - e_2 \cdot \overline{e_1})/\sqrt 2   \nonumber\\
u_3 &=& e_1 \cdot \overline{e_1} - e_2 \cdot \overline{e_2}/\sqrt 2   \nonumber
\end{eqnarray}

Then we have
\[
Ju_j = u_j
\]
for all $j$ and so $u_j \in V$ for all $j$.   We also have

\begin{eqnarray*}
g(u_0, u_0) &=& g(e_1 \cdot \overline{e_1} + e_2 \cdot \overline{e_2}, e_1 \cdot \overline{e_1} + e_2 \cdot \overline{e_2})\\
&=& \frac{\varepsilon(e_1,e_2)\overline{\varepsilon}(\overline{e_1}, \overline{e_2})}{2} + \frac{\varepsilon(e_2,e_1) \varepsilon(\overline{e_2}, \overline{e_1})}{2}\\
&=& \frac{1}{2} + \frac{1}{2}\\
&=& 1
\end{eqnarray*}

Similarly, we have 
\begin{eqnarray*}
g(u_j,u_j) &=& -1 \quad \quad 1 \leq j \leq 3\\
g(u_i,u_j) &=& 0  \quad \quad i\neq j, \: 0 \leq i,j \leq 3
\end{eqnarray*}

so that $g$ determines a quadratic form $Q$ of rank $4$ and signature $2$, which we will denote as follows:
\[ 
+ ---
\]
This is the Lorentz quadratic form.     
\end{proof}

It is a truly remarkable fact that the Lorentz quadratic form emerges from an apparently symmetrical situation. 

The space $S \oplus \overline{S}$ is a complex vector space of dimension $4$.   It is called the space of $4$-\emph{spinors}.
We will write
\[
\fS:  = S \oplus \overline{S}.
\]

\subsection{$4$-spinors as a Clifford module}  Thanks to Paul Robinson for help with this section.   

Let $a,b,c \in S$.   Since $S$ has dimension $2$, there must be a linear relation among $a,b,c$.   We make this precise.

\begin{lem}\cite[1.6.19]{PR}.  Let $a,b,c \in S$.   Then we have the cyclic identity
\begin{eqnarray}\label{cyclic}
\varepsilon(b,c)a + \varepsilon(c,a)b + \varepsilon(a,b)c = 0
\end{eqnarray}
\end{lem}

\begin{proof}   Fix $s \in S$ and define
\[
f(a,b,c): = \varepsilon(s, \varepsilon(b,c)a + \varepsilon(c,a)b + \varepsilon(a,b)c).
\]
Then $f$ is a skew $3$-form on a $2$-space, hence vanishes for all $s \in S$.  Now invoke the non-singularity of $\varepsilon$.   
\end{proof}

Define a map 
\[
\phi : S \otimes \overline{S} \to \mathrm{End}(\fS)
\]
as follows:
\[
\phi(p \otimes \overline{q})(a \oplus \overline{b}): = \sqrt 2 \left[ \overline{\varepsilon}(\overline{b}, \overline{q}) p \oplus \varepsilon(p,a)\overline{q})\right]
\]

\begin{lem}\label{quadratic} \cite[\S 2.3]{PW}. For all $X, Y \in S \otimes \overline{S}$, we have
\[
\phi(X)\phi(Y) + \phi(Y)\phi(X) = h(X,Y)I_{\fS}
\]
\end{lem}

\begin{proof}

We have
\begin{eqnarray*}
 \frac{1}{2} \phi(r \otimes \overline{s})\phi(p \otimes \overline{q})(a \oplus \overline{b}) 
=&  \frac{1}{\sqrt 2} \phi(r \otimes \overline{s})\left[\overline{\varepsilon}(\overline{b}, \overline{q}) p \oplus \varepsilon(p,a)\overline{q}  \right]  \\
=&    \overline{\varepsilon}(\varepsilon(p,a) \overline{q}, \overline{s})r  \oplus \varepsilon(r, \overline{\varepsilon}(\overline{b}, \overline{q})p) \overline{s} \\
=& \varepsilon(p,a)\overline{\varepsilon}(\overline{q}, \overline{s})r  \oplus \varepsilon(r,p) \overline{\varepsilon}(\overline{b}, \overline{q})\overline{s}
\end{eqnarray*}

Next, we symmetrize and add:
\begin{eqnarray*}
& &
\frac{1}{2} \{\phi(p \otimes \overline{q}) \phi(r \otimes \overline{s}) + \phi(r \otimes \overline{s}) \phi(p \otimes \overline{q}) \}( a \oplus \overline{b}) \\
& = & \varepsilon(p,a)\overline{\varepsilon}(\overline{q}, \overline{s})r \oplus \varepsilon(r,p)\overline{\varepsilon}(\overline{b}, \overline{q})\overline{s} 
 + \varepsilon(r,a)\overline{\varepsilon}(\overline{s}, \overline{q})p \oplus \varepsilon(p,r)\overline{\varepsilon}(\overline{b}, \overline{s})\overline{q} \\
& = & \overline{\varepsilon}(\overline{q}, \overline{s}) [ \varepsilon(p,a)r - \varepsilon(r,a)p ] \oplus \varepsilon(p,r)[ \overline{\varepsilon}(\overline{b}, \overline{s}) \overline{q} - \overline{\varepsilon}(\overline{q}, \overline{b}) \overline{s} ]
\end{eqnarray*}

Applying the cyclic identity (\ref{cyclic}), we obtain
\begin{eqnarray*}
\varepsilon(p,a)r - \varepsilon(r,a)p  &=&  \varepsilon(p,r)a \\
\overline{\varepsilon}(\overline{b}, \overline{s})\overline{q} - \overline{\varepsilon}(\overline{q}, \overline{b}) \overline{s} & = & \overline{\varepsilon}(\overline{q}, \overline{s}) \overline{b}
\end{eqnarray*}

Our conclusion is that
\begin{eqnarray*}
\frac{1}{2} \{\phi(p \otimes \overline{q}) \phi(r \otimes \overline{s}) + \phi(r \otimes \overline{s}) \phi(p \otimes \overline{q}) \}( a \oplus \overline{b}) 
& = & \varepsilon(p,r) \overline{\varepsilon}(\overline{q}, \overline{s})( a \oplus \overline{b})
\end{eqnarray*}
as required.
\end{proof}

 
 In view of Lemma \ref{quadratic}, the map $\phi$ will lift to a morphism of $\C$-algebras
\[
\mathcal{C}\ell(V,Q)\otimes_{\R} \C \to \textrm{End}(\fS)
\]
where $\mathcal{C}\ell(V,Q)$ is the Clifford algebra of $V$ with respect to the quadratic form $Q$.      Therefore, $\fS$ slots in as a pointwise irreducible $\Z/2\Z$-graded Clifford module as in \cite{ABS}.

\subsection{Momentum space}  Let $x \in V$.   The following equation defines a character of $V$:
\[
\widehat{x}(y) = \exp(ig(x,y))
\]
and the map
\begin{eqnarray}\label{duality}
V \to \widehat{V},  \quad \quad x \mapsto \widehat{x}
\end{eqnarray}
secures a (non-canonical) isomorphism of the real vector space $V$ onto its Pontryagin dual $\widehat{V}$.   The vector space $\widehat{V}$ admits an $\SL_2(\C)$-action as follows:
\[
(A \cdot \widehat{x})(y) = \widehat{x}(A^{-1} y)
\]

The isomorphism $V \to \widehat{V}$ commutes with the action of $\SL_2(\C)$.   

We recall the basis $\{u_0, u_1, u_2, u_3\}$ determined by the symplectic basis $\{e_1, e_2\}$, see (\ref{world}).       A  basis of $\widehat{V}$ now presents itself, namely
\[
v_j : = \widehat{u_j}
\]

As explicit characters, we have
\[
v_0(x) = \exp(ix_0), \: v_1(x) = \exp(- ix_1), \: v_2(x) = \exp(- ix_2), \: v_3(x) = \exp( - i x_3)
\]

The basis $\{u_o, u_1, u_2, u_3\}$ determines an isomorphism
\[
V \to \R^4, \quad \quad x \mapsto (x_0, x_1, x_2, x_3)
\]
and the basis $\{v_0, v_1, v_2, v_3\}$ determines an isomorphism
\[
\widehat{V} \to \R^4, \quad \quad p \mapsto (p_0, p_1, p_2, p_3).
\]

We will write 
\[
\P^4 = \widehat{V}.
\]


Now $V$ equipped with its Lorentz quadratic form is isomorphic to $\mathbb{P}^4$ with its quadratic form 
\[
Q(p) = p_0^2 - p_1^2 - p_2^2 - p_3^2
\]
We record this in the following
\begin{lem}\label{commute}
Thanks to the map (\ref{duality}), we have 
\[
(V,Q) \simeq (\P^4,Q)
\]
 and this map commutes with the action of $\SL_2(\C)$.   
 \end{lem}
 
Lemma \ref{commute} allows us, by transport of structure, to obtain the following

\begin{cor}  For all $x,y \in \mathbb{P}^4$ we have a morphism
\[
\phi : \mathcal{C}\ell(\mathbb{P}^4, Q) \to \textrm{End}(\fS)
\]
for which 
\[
\phi(x)\phi(y) + \phi(y)\phi(x) = 2g(x,y)
\]

In particular, we have
\[
\phi(p)^2 = Q(p)
\]
for all $p \in \mathbb{P}^4$.  
\end{cor}

\begin{defn}\label{pi}    Let $A \in \SL_2(\C)$.   The canonical representation $\pi$ of $\SL_2(\C)$ on $S \otimes \overline{S}$ is defined, on elementary tensors, by the equation
\[
 \pi(A): = A \otimes \overline{A}, \quad \quad \quad (x \otimes \overline{y}) \mapsto Ax \otimes \overline{A} \overline{y}
\]

The canonical representation $\tau$ of $\SL_2(\C)$ on $\fS = S \oplus \overline{S}$ is defined, on elementary tensors, by the equation
\[
\tau(A): = A \oplus \overline{A}, \quad \quad \quad (x \oplus \overline{y}) \mapsto Ax \oplus \overline{A}\overline{y}
\]
\end{defn}

The representations $\pi$  and $\tau$ are intimately related to the Clifford module map $\phi$ in the following way.

\begin{thm}\label{Cliff}   For all $A \in \SL_2(\C), X \in V$, we have 
\[
\phi(\pi(A)(X)) = \tau(A) \phi(X) \tau(A^{-1})
\]
\end{thm}

\begin{proof} Let $X = p \otimes \overline{q}$.   Then the LHS is 
\begin{eqnarray*}\label{abc}
\phi(\pi(A)(X))( a \oplus \overline{b}) &=& \phi(Ap \otimes \overline{A}\overline{q})(a \oplus \overline{b}) \\
& = & \overline{\varepsilon} (\overline{b}, \overline{A} \overline{q})Ap \oplus \varepsilon(Ap,a)\overline{A}\overline{q} 
\end{eqnarray*}

We also have
\begin{eqnarray*}
\phi(p \otimes \overline{q}) \tau(A^{-1})(a \oplus \overline{b}) &=& \phi(p \otimes \overline{q})(A^{-1}a \oplus \overline{A}^{-1}\overline{b}) \\
& = & \overline{\varepsilon} (\overline{A}^{-1} \overline{b}, \overline{q})p \oplus (p, A^{-1}a) \overline{q}
\end{eqnarray*}
so that the RHS is 
\begin{eqnarray*}
\tau(A) \phi(p \otimes \overline{q}) \tau(A^{-1})(a \oplus \overline{b}) &=& \overline{\varepsilon} (\overline{A}^{-1} \overline{b}, \overline{q})Ap \oplus (p, A^{-1}a) \overline{A}\overline{q}\\
& = & \overline{\varepsilon} (\overline{b}, \overline{A} \overline{q})Ap \oplus \varepsilon(Ap,a)\overline{A}\overline{q}
\end{eqnarray*}
which is the LHS.      Note that we always have 
\[
\varepsilon(Ax,Ay) =  \det(A) \cdot \varepsilon(x,y) = \varepsilon(x,y)
\]
 for all $A \in \SL_2(\C)$.   
\end{proof}

 \begin{lem}
 The restriction of  $\pi$ to $V$  preserves the Lorentz quadratic form $Q$.
 \end{lem} 
 \begin{proof}   The map $\pi(A)$ commutes with $J$:
\begin{eqnarray*}
\pi(A)J (x \otimes \overline{y}) &=&  \pi(A)(y \otimes \overline{x})\\
&=& Ay \otimes \overline{A} \overline{x}\\
&=& \pi(A)(y \otimes \overline{x})\\
&=& \pi(A)J(x \otimes \overline{y})
\end{eqnarray*}
for all $A \in \SL_2(\C)$ and all elementary tensors $x \otimes \overline{y}$.   The map $\pi(A)|_V$ is well-defined because 
\[
Jv = v \implies \pi(A)v = \pi(A)Jv = J\pi(A)v
\]
for all $A \in \SL_2(\C)$ and all $v \in V$.   Finally, we have
\begin{eqnarray*}
h(\pi(X), \pi(Y)) & = & h(Ap \otimes \overline{A}\overline{q}, Ar \otimes \overline{A}\overline{s})\\
&=&   \varepsilon(Ap, Ar) \overline{\varepsilon}(\overline{A} \overline{q}, \overline{A} \overline{s})\\
&=& \varepsilon(p, r) \overline{\varepsilon}(\overline{q}, \overline{s})\\
&=& h(X,Y).
\end{eqnarray*}   
 \end{proof}

 We have a commutative diagram
 \[
 \begin{CD}
 V @>>> S \otimes \overline{S}\\
 @V {\pi(A)|_V} VV               @VV {\pi(A)} V\\
 V @>>> S \otimes \overline{S}
  \end{CD}
 \]
 in which the horizontal maps are  canonical inclusions and the left vertical map is a Lorentz transformation.

\section{The Dirac Equation}  

\subsection{The Dirac Bundle}  We consider one of the orbits  of $\SL_2(\C)$ acting on $\P^4$.   Let $m > 0$ and let $v_0 = (1,0,0,0) \in \P^4$.     Explicitly,
$v_0$ is the character of $\R^4$ given by $x \mapsto \exp(ix_0)$.   Define \[
p_0 = mv_0 \in \P^4
\]
and define $X_m$ to be the orbit of $p_0$:  
\[
X_m: = \SL_2(\C) \cdot p_0 \subset \P^4
\]
This orbit is \emph{the mass shell} associated with the positive mass $m$.   

We recall that
\[
\mathfrak{S} =  S \oplus \overline{S}.
\]
We begin with the following trivial vector bundle over $X_m$:
\[
X_m \times \fS \to X_m, \quad \quad \quad (p, \Psi) \mapsto   p
\]
With $A \in \SL_2(\C)$, we define
\begin{eqnarray*}
A \cdot p &=& \pi(A) p\\
A \cdot \Psi &=& \tau(A)\Psi\\
A \cdot (p, \Psi) &=& (A \cdot p, A \cdot \Psi)
\end{eqnarray*}
We obtain a trivial $\SL_2(\C)$-bundle of rank $4$, i.e. the total space and the 
base space admit an action of $\SL_2(\C)$ which commutes with the projection.

 We now construct the \emph{Dirac bundle} as a sub-bundle.     The total space is
\[
\mathfrak{B}_m : = \{(p,\Psi) : p \in X_m, \; \Psi \in \fS, \; \phi(p)  \Psi  = m\Psi\}
\]
the base space is $X_m$, and the projection is 
\[
\fB_m \to X_m, \quad \quad (p,\Psi) \mapsto p.
\]

The fibre at $p \in X_m$ is the linear subspace of $\fS$ given by 
\[
\{(p ,\Psi) : \phi(p) \cdot \Psi = m\Psi\}
\]

Now we have
\begin{eqnarray*}
\phi(p) \Psi &=& m\Psi\\
\implies \tau(A) \phi(p) \Psi &=& m \tau(A) \Psi\\
\implies \tau(A) \phi(p) \tau(A^{-1}) \tau(A)\Psi &=& m \tau(A) \Psi\\
\implies \phi(A p) (\tau(A) \Psi) &=& m \tau(A)  \Psi
\end{eqnarray*}
by Theorem (\ref{Cliff}).   That is to say,
\[
\phi(p) \Psi = m\Psi \implies  \phi(A p) ( A \cdot \Psi) = m(A \cdot \Psi)
\]

This shows  that $\SL_2(\C)$ sends, for each $A \in \SL_2(\C)$,  the fibre at $p$ to the fibre at $A \cdot p$.   Therefore,  the Dirac bundle $\fB_m \to X_m$ is an $\SL_2(\C)$-bundle.

We have $\phi(p_0) = m \phi(v_0) = m \gamma_0$ and so the fibre at $p_0$   is given by
\[
\{ \Psi \in \fS : \gamma_0 \Psi = \Psi \}
\]

This is the eigenspace of $\gamma_0$ with eigenvalue $1$.   Now $\gamma_0^2 = 1$, so $\gamma_0$ has two eigenspaces $\fS_+, \fS_-$ of dimension $2$ with eigenvalues $+1, -1$.  

Therefore, the Dirac bundle is an $\SL_2(\C)$-bundle of rank $2$.

The isotropy subgroup of $\SL_2(\C)$ at $p_0$ is the compact Lie group $\SU_2$.    By Lemma \ref{commute} and Theorem \ref{Cliff}, we have 
\[
\tau(A) \gamma_0 = \gamma_0 \tau(A)
\]
for all $A \in \SU_2$.   We have
\[
\Psi \in \fS_+ \implies \gamma_0 \Psi = \Psi \implies \tau(A) \gamma_0 \Psi = \tau(A)\Psi \implies \gamma_0 \tau(A) \Psi = \tau(A) \Psi \implies \tau(A) \Psi \in \fS_+
\]
so that $\tau |_{\SU_2}$ leaves $\fS_+$ invariant.

The vector space $\fS$ admits a natural basis, namely $\{e_1, e_2, \overline{e_1}, \overline{e_2}\}$.  We calculate that
\begin{eqnarray*}  \phi(v_0) (e_2 + \overline{e_1}) &=& e_2 + \overline{e_1}\\ 
\phi(v_0) (e_1 - \overline{e_2}) &=& e_1 - \overline{e_2}\\ 
\phi(v_0) (e_1 + \overline{e_2}) &=& -(e_1 + \overline{e_2})\\
\phi(v_0)(e_2 - \overline{e_1}) &=& - (e_2 - \overline{e_1}) \end{eqnarray*} The $+1$-eigenspace of $\phi(v_0)$ is therefore given by
\[
\fS_+ = \textrm{span of} \: \{e_1 - \overline{e_2}, e_2 + \overline{e_1}\}
\]
  
Let \[A = \left(\begin{array}{cc}e^{it} & 0\\0 & e^{-it}\end{array}\right)\]Then we have
\[\tau(A)e_1 = e^{it} e_1,\quad  \tau(A)e_2 = e^{-it}e_2, \quad \tau(A)\overline{e_1} = e^{-it}\overline{e_1}, \quad \tau(A)\overline{e_2} = e^{it}\overline{e_2}\]
\begin{eqnarray*}\tau(A)(e_1 - \overline{e_2}) &=& e^{it} (e_1 - \overline{e_2})\\\tau(A)(e_2 + \overline{e_1}) &=& e^{-it} (e_2 + \overline{e_1})
\end{eqnarray*}The character of $\tau(A)$ is therefore $e^{it} + e^{-it}$ and so $\tau |_{\SU_2}$ is  the spin $1/2$ representation of $\SU_2$ on $\fS_+$.

\subsection{The Dirac bundle as an associated vector bundle}  

   Let $p_0 = (m,0,0,0)$ and note that
   \[
   \SL_2(\C)/\SU_2 \simeq X_m, \quad \quad A \mapsto Ap_0.
   \]
   
   We construct the associated vector bundle
   \[
   \SL_2(\C) \times _{\SU_2} \fS_+
   \]
  where $\SU_2$ acts on $\fS_+$ via the spin $1/2$ representation $\tau |_{\SU_2}$.    The elements of the associated vector bundle are equivalence classes $[A,\Psi]$ 
  with $A \in \SL_2(\C)$ and $\Psi \in \fS_+$, defined as
\[
[A,\Psi]: = \{(AT, \tau(T)^{-1}\Psi):  T \in \SU_2 \}
\]

\begin{thm}\label{SSS}   There is an isomorphism of $\SL_2(\C)$-bundles as follows:
\[
\beta : \SL_2(\C) \times_{\SU_2} \fS_+ \simeq \fB_m
\]
\[
[A,\Psi] \mapsto (Ap_0, \tau(A)\Psi)
\]   
\end{thm}

\begin{proof}   Note that $\fS_+ \subset \fS$ and $\tau(A)\Psi \in \fS$.   The map $\beta$ is well-defined: replacing $A$ by $AT$ and  $\Psi$ by $\tau(T^{-1})\Psi$, we have
\[
(AT, \tau(T)^{-1}\Psi) \mapsto (ATp_0, \tau(AT) \tau(T)^{-1}\Psi) = (Ap_0, \tau(A)\Psi).
\] 

We have
\begin{eqnarray*}
 \phi(Ap_0)\tau(A)\Psi &=& \tau(A) \phi(p_0)\tau(A^{-1})\tau(A)\Psi \quad \textrm{by} \quad\textrm{Theorem}\; \ref{Cliff}\\
&=& \tau(A) \phi(p_0) \Psi\\
&=& \tau(A) m \gamma_0 \Psi\\
&=& \tau(A) m\Psi \quad \textrm{since} \quad \Psi \in \fS_+\\
&=& m \tau(A)\Psi
\end{eqnarray*}
so that  $(Ap_0, \tau(A)\Psi) \in \fB_m$ as required.

Conversely, given $(p,\Psi) \in \fB_m$.   
Choose $A$ such that $p = Ap_0$ and set $\Phi = \tau(A)^{-1} \Psi$.    Then we have 
\begin{eqnarray*}
\phi(p)\Psi &=& m\Psi\\
\implies \tau(A) \phi(p_0) \tau(A)^{-1} \tau(A)\Phi  &=& m \tau(A)\Phi\\
\implies \tau(A) \phi(p_0) \Phi &=& m \tau(A)\Phi\\
\implies  \gamma_0 \Phi &=& \Phi\\
\implies \Phi &  \in & \fS_+
\end{eqnarray*}

Therefore, we have 
\[
\beta: (A,\Phi) \mapsto (Ap_0,\tau(A)\Phi) = (p,\Psi)
\]
as required.   

Now $A$ is determined mod $\SU_2$.  If we replace $A$ by $AT$ with $T \in \SU_2$ then we must replace $(A,\Phi)$ by
$(AT,y)$  with $y$ defined by 
\[
y = \tau(AT)^{-1}\Psi = \tau(T)^{-1} \tau(A)^{-1}\Psi = \tau(T)^{-1}\Phi \in \fS_+
\]
This leads to the pair 
\[
(AT, \tau(T)^{-1}\Phi)
\]
which is in the same equivalence class as $(A,\Phi)$.   Therefore, the inverse of the map $\beta$ is well-defined.
\end{proof}

\begin{cor}\label{dirac}
The Dirac bundle is isomorphic as an $\SL_2(\C)$-bundle to the associated vector bundles
\[
\SL_2(\C) \times_{\SU_2} S
\]
and
\[
\SL_2(\C) \times_{\SU_2} \overline{S}
\]

\end{cor}

\begin{proof}   There is a natural action of $\SL_2(\C)$ on $S$.  The action on the conjugate vector space is defined as
\[
A \cdot \overline{v} = \overline{Av}
\]
with $A \in \SL_2(\C)$ and $v \in S$.   Now use the fact that $\fS_+, S$ and $\overline{S}$ are isomorphic as $\SU_2$-modules.
\end{proof}

 
\subsection{Dirac equation}  The condition
\[
\phi(p)  \Psi = m\Psi
\]
with $\Psi \in \fS$ is, written in full,
\[
\left(p_0 \gamma_0 + p_1 \gamma_1 + p_2\gamma_2 + p_3\gamma_3 \right) \Psi = m\Psi
\]

In the slashed notation \cite[p.118]{F}, this equation appears as
\begin{eqnarray}\label{slashed}
\slashed{p} \Psi = m\Psi
\end{eqnarray}
This is the Dirac equation.  After Fourier transform of the variables $p_0, p_1, p_2, p_3$, the equation appears in the more familiar form
\[
\sum_{r=0}^3 i\gamma_r \frac{\partial}{\partial x_r} \Psi = m\Psi
\]
where    $\gamma_\mu : = \phi(v_\mu)$ with $\mu = 0,1,2,3,$ are $\gamma$-matrices satisfying the conditions
\begin{eqnarray*}
\gamma_0^2 &=& 1\\
 \gamma_1^2 & =&  \gamma_2^2  \: =  \: \gamma_3^2 \: = -1\\
  \gamma_\mu \gamma_\nu & =&  - \gamma_\nu \gamma_\mu
\end{eqnarray*}
for all $\mu \neq \nu$. 

\subsection{A pair of conjugate $2$-spinor fields}   Let $\Psi$ be a solution of the Dirac equation (\ref{slashed}), i.e. $\phi(p) \Psi = m\Psi$.   To mark the dependence on $p$, we will write
this equation as 
\[
\phi(p) \Psi_p = m \Psi_p
\]
Then $\{\Psi_p : p \in X_m\}$ is a section of the Dirac bundle $\fB_m$.   By Corollary \ref{dirac}, $\fB_m$ is isomorphic to the associated bundle 
\begin{eqnarray}\label{B1}
\SL_2(\C) \times_{\SU_2} S
\end{eqnarray}
The section $\{\Psi_p : p \in X_m\}$ therefore determines a section of (\ref{B1}), which is a $2$-spinor field over the manifold $X_m$.      Now $\fB_m$ is also isomorphic to the conjugate bundle
\begin{eqnarray}\label{B2}
\SL_2(\C) \times_{\SU_2} \overline{S}
\end{eqnarray}
and so the section $\{\Psi_p : p \in X_m\}$ determines a section of (\ref{B2}), which is the conjugate $2$-spinor field.   In this way, a solution of the Dirac equation determines a pair of conjugate
$2$-spinor fields over the manifold $X_m$.

\subsection{Comments}   It is worth noting that the Clifford module structure of $\fS$ survives up to and including the Dirac equation, for 
$\phi(p)$ belongs to the Clifford algebra $\mathcal{C} \ell (\P^4,Q)$.   

In the literature, in the mathematical accounts of the Dirac equation \cite[Eqn.(136)]{Var}, \cite[p.72]{Sim}, \cite[Eqn.1.3.52]{Jos} one finds the 
map
\[
A \mapsto A \oplus (A^*)^{-1}
\]
instead of our map 
\[
\tau(A) =  A \oplus \overline{A}
\]
From a conceptual point of view, the adjoint $A^*$ of $A \in \SL(S)$ is not defined, because $S$ is not equipped with a sesquilinear inner product.   The space $S$ is endowed only with the symplectic
form $(x,y) \mapsto \varepsilon(x,y)$.   

\section{The transition to spinor fields on space-time}   So far, we have dealt only with $2$-spinors.  In order to apply this theory to \emph{spinor fields} on space-time,  we proceed as follows.

Let $\fM$ be a smooth real $4$-manifold.   The first observation is that we can view  $S$ as a (smooth) complex $2$-plane  
bundle over the base space $\fM$.   The second observation is that the operations of conjugation, direct sum and tensor product  are as easy to apply to vector bundles as to vector spaces.
Then $\overline{S}$ is the conjugate of $S$, $\varepsilon$ is a non-degenerate symplectic form on $S$, $V$ is the real part of $S \otimes \overline{S}$,
and $g$ is the Lorentz metric with signature $+ - - -$ given by
\[
g = \varepsilon \otimes \overline{\varepsilon}.
\]

  Then $V$ is a real $4$-plane bundle on $\fM$, equipped with the Lorentz metric tensor.    At this point, we have to identify $V$ with the tangent bundle $T\fM$ of $\fM$:
\[
T\fM = V.
\]
In the terminology of \cite[p.211]{PR}, we must identify the world-vectors in $V$ with the tangent vectors to the  manifold $\fM$.   
In that case,  we can view $\fM$ as a model of curved space-time in general relativity.   The smooth sections of the vector bundle $S$ are called $2$-spinor fields; the smooth sections of 
$\fS$ are called $4$-spinor fields.  

The whole of \S2 carries over, \emph{at no extra cost}, to the theory of spinor fields on manifolds.   In this formulation, the curved space-time $\fM$ of general relativity is
 \emph{already equipped} with spinor fields.

\end{document}